\documentclass[12pt]{article}
\usepackage{mathpazo}      
\usepackage{amssymb,amsmath,amsthm,amstext,graphics,amsfonts,hyperref, xcolor, stmaryrd, mathtools}
\usepackage{cite}
\usepackage{graphicx}
\usepackage{enumerate}
\usepackage[justification=centering]{caption}
\usepackage[labelformat=simple]{subcaption}
\usepackage{enumitem}

\newcommand{\FF}{{\mathbb{F}}}
\newcommand{\seq}{\subseteq}
\newcommand{\ZZ}{\mathbb{Z}}

\newcommand{\cv}{\mathbf{c}}

\newcommand{\ev}{\mathbf{e}}

\newcommand{\xv}{\mathbf{x}}
\newcommand{\yv}{\mathbf{y}}
\newtheorem{theorem}{Theorem}[section]
\newtheorem{lem}[theorem]{Lemma}
\newtheorem{cor}[theorem]{Corollary}
\newtheorem{prop}[theorem]{Proposition}
\theoremstyle{definition}
\newtheorem{ex}[theorem]{Example}

\newtheorem{rem}{Remark}

\usepackage{mathrsfs}

\begin{document}
\title{A general family of Plotkin-optimal two-weight codes over $\ZZ_4$}

\author{Hopein Christofen Tang
\footnote{Combinatorial Mathematics Research Group, Faculty of Mathematics and Natural Sciences, Institut Teknologi Bandung, Jl. Ganesha 10, Bandung, 40132, INDONESIA,\hfill \texttt{hopeinct@students.itb.ac.id}}
	 ~~and Djoko Suprijanto
\footnote{Combinatorial Mathematics Research Group,
 Faculty of Mathematics and Natural Sciences,
 Institut Teknologi Bandung,
 Jl. Ganesha 10, Bandung, 40132,
 INDONESIA,\hfill \texttt{djoko.suprijanto@itb.ac.id} }
}

\maketitle

\begin{abstract}
We obtained all possible parameters of Plotkin-optimal two-Lee weight projective codes over $\ZZ_4,$ together with their weight distributions.  We show the existence of codes with these parameters as well as their weight distributions by constructing an infinite family of two-weight codes. Previously known codes constructed by
Shi et al. (\emph{Des Codes Cryptogr.} {\bf 88}(12): 2493-2505, 2020) 
can be derived as a special case of our results.
We also prove that the Gray image of any Plotkin-optimal two-Lee weight projective codes over $\ZZ_4$ has the same parameters and weight distribution as some two-weight binary projective codes of type SU1 in the sense of Calderbank and Kantor (\emph{Bull. Lond. Math. Soc.} {\bf 18}:97-122, 1986).
\end{abstract}

\vspace{0.5cm}

\vspace{0.5cm}

\section{Introduction}
A code $C \seq \FF^n$ over an alphabet $\FF$ is called a constant weight code (or one-weight code) if every nonzero codeword has the same weight.  Well-known examples of such codes over $\FF_q$ are simplex codes, which are the duals of Hamming codes (see, e.g.,\cite{Huffman}).  Constant weight codes over $\FF_q$ have been extensively studied (see \cite{Sloane1977} for a detailed account). One important result regarding constant weight linear codes over $\FF_2$ is that for every positive integer $k$, there exists a unique (up to equivalence) one-weight code of dimension $k$ such that any two columns in its generator matrix are linearly independent. A similar phenomenon was proved by Carlet \cite{Carlet2000} for constant weight linear codes over $\ZZ_4.$

A two-weight code $C \seq \FF^n$ over an alphabet $\FF$ is a code with exactly two nonzero different weights. Two-weight codes over $\FF_q$ have been investigated since a long time because of their connection to other combinatorial objects such as strongly regular graphs, difference sets and also finite geometries \cite{Calderbank1982, Calderbank1986, Delsarte1972}. Recently, there were also some progress in the study of two-weight codes over finite rings, such as $\ZZ_4$ \cite{Shi2014, Shi2017, Shi2020}, $\ZZ_{2^m}$ \cite{Li2022, Shi2018, Shi2021}, $\ZZ_{p^m}$ \cite{Shi2021}, and $\FF_p+u\FF_p$ \cite{Shi2017-2}, with $u^2=u$, where $p$ is an odd prime. Interestingly, two-weight codes over $\ZZ_{p^m}$ for some prime $p$ also have some connections to strongly regular graphs \cite{Shi2018} and finite geometries \cite{Shi2021}.

Constructing (optimal) codes is one of the main problems in coding theory. The construction of one-weight codes over $\ZZ_4$ has been provided in \cite{Shi2014}, while several constructions of two-weight codes over $\ZZ_4$ have been found by several authors \cite{Shi2018, Shi2014, Shi2017}, and recently by \cite{Shi2020}. All two-weight codes over $\ZZ_4$ constructed in \cite{Shi2018, Shi2014, Shi2017, Shi2020} are Plotkin-optimal and projective (see Section 2 for the definition). The purpose of this paper is to find a complete list of possible parameters and weight distributions of Plotkin-optimal two-weight projective codes over $\ZZ_4$.

The organization of the paper is as follows.  In Section 2, we provide some definitions related to linear codes over $\ZZ_4.$ Several useful identities related to linear codes over $\ZZ_4$ are derived in Section 3. In Section 4, we study possible parameters and weight distributions of Plotkin-optimal two-weight projective codes over $\ZZ_4$ and their Gray images. In Section 5, we first consider a family of free two-weight codes over $\ZZ_4$. After that, we use these free codes to construct an infinite family of two-weight codes. The paper is ended by concluding remarks.
We refer the readers to \cite{Huffman, Sloane1977} for undefined terms in coding theory.




\section{Preliminaries}

A \emph{code} $C$ of length $n$ over the ring $\ZZ_4$ is a nonempty subset of $\ZZ_4^n.$ If the code is also a submodule of $\ZZ_4^n,$
then we say that the code is \emph{linear.} The linear code is called \emph{free} if it is a free submodule of $\ZZ_4^n.$

A matrix $G\in \ZZ_4^{k \times n}$ is called a \emph{generator matrix} of a linear code $C$ of length $n$ over $\ZZ_4$ if the rows of $G$ generate $C$ and no proper subset of the rows of $G$ generates $C$.

Two codes are said to be \emph{equivalent} if one can be obtained from the other by permuting the coordinates and (if necessary) changing the signs of certain coordinates. Codes differing by only a permutation of coordinates are called \emph{permutation-equivalent.} It is well-known (see \cite{Hammons1994}) that any linear code over $\ZZ_4$ is permutation-equivalent to the linear code $C$ with generator matrix $G$ of the form
\begin{equation}\label{G-standard}
	G=
	\begin{pmatrix}
		I_{k_1} & A & B_1+2 B_2\\
		0 & 2 I_{k_2} & 2 D
	\end{pmatrix},
\end{equation}
where $A,B_1,B_2,$ and $D$ are $(0,1)$-matrices. Moreover, the code $C$ is a free linear code if and only if $k_2=0.$ The generator matrix of a linear code $C$ over $\ZZ_4$ is called in a \emph{standard form} if it has the form as given in the Equation (\ref{G-standard}).

The \emph{Lee weight} of $ x \in \ZZ_4,$ denoted by $w_L(x),$ is defined by $w_L(0)=0,$ $w_L(1)=1,$ $w_L(2)=2,$ and $w_L(3)=1.$ The Lee weight of a vector $\xv=(x_1,x_2,\ldots,x_n) \in \ZZ_4^n$ is defined as
$
w_L(\xv)=\sum_{i=1}^n w_L(x_i)$. For a given positive integer $\beta,$ a linear code $C$ over $\ZZ_4$ is called a \emph{$\beta$-weight code} if
\[
|\{w_L(\cv)\;\colon\;\mathbf{0}\neq \cv\in C\}|=\beta,
\]
where $\mathbf{0}$ denotes the all-zero vector.

For $\xv,\yv \in \ZZ_4^n,$ the \emph{Lee distance} between $\xv$ and $\yv,$ denoted by $d_L(\xv,\yv),$ is defined by
$d_L(\xv,\yv)=w_L(\xv-\yv).$ The \emph{minimum Lee distance} of a linear code $C \seq \ZZ_4^n$ is defined by
\[
d_L=d_L(C):=\mathrm{min}\{d_L(\xv,\yv):~\xv,\yv \in C, \xv \neq \yv\}.
\]
Throughout the paper, what we mean by weight and distance are Lee weight and Lee distance, respectively. It is clear that the minimum Lee distance of a linear code $C$ is exactly the same as the minimum Lee weight, namely
$d_L(C)=\mathrm{min}\{ w_L(\xv):~\mathbf{0} \neq \xv\in C\}.$
We write the parameters of a linear code $C$
over $\ZZ_4$ as $[n, 4^{k_1}2^{k_2}, d_L],$ where $n$ is a length of $C,$ $|C| = 4^{k_1}2^{k_2},$ and $d_L = d_L(C).$ Following
Hammons et al. \cite{Hammons1994} (cf. \cite{Wan}), we say that $C$ is of type $4^{k_1}2^{k_2}.$ In \cite{Hammons1994}, the \emph{Gray map} $\phi: \ZZ_4\rightarrow \FF_2^2$ is defined as $\phi(0)=(0,0), \phi(1)=(0,1), \phi(2)=(1,1)$, and $\phi(3)=(1,0)$. The Gray map $\phi$ can be extended naturally to $\Phi:\ZZ_4^n \rightarrow \FF_2^{2n}$. $\Phi$ is a distance-preserving map (isometry) from $(\ZZ_4^n,d_L)$ to $(\FF_2^{2n},d_H)$, where $d_H$ denotes the Hamming distance.

In 1968, Wyner and Graham \cite{Wyner1968} proved an upper bound for the minimum Lee distance of codes over $\ZZ_4$, as follows (see also \cite[Theorem 3.4]{Hopein1} for a simpler proof).

\begin{theorem}[Plotkin-type Lee distance bound]\label{Plotkin}
	If $C$ is a linear code of length $n$ over $\ZZ_4$ with parameters $[n,4^{k_1}2^{k_2},d_L]$, then
	\begin{equation}\label{PLDB}
		d_L\leq \frac{|C|}{|C|-1}n.
	\end{equation}
\end{theorem}

A linear code $C$ whose minimum Lee distance $d_L(C)$ is an integer nearest to the upper bound of the Plotkin-type Lee distance bound (\ref{PLDB}) as given in Theorem \ref{Plotkin} is called \emph{Plotkin-optimal.} In other words, a linear code $C$ is Plotkin-optimal if $\displaystyle d_L(C)=\left \lfloor \frac{|C|}{|C|-1}n \right \rfloor.$

For $\xv=(x_1,x_2,\ldots,x_n),~\yv=(y_1,y_2,\ldots,y_n)\in\ZZ_4^n,$ we define the \emph{Euclidean inner product} of $\xv$ and $\yv$ as
$\xv\cdot\yv=\sum_{i=1}^n x_i y_i\in \ZZ_4.$ The \emph{Dual} of $C,$ denoted by $C^\perp,$ is defined as
\[
C^\perp=\{\xv \in \ZZ_4^n:~\xv \cdot \yv=0, \text{ for all }\yv \in C\}.
\]
A \emph{projective code} $C$ of length $n$ over $\ZZ_4$ is a linear code such that the minimum Lee distance of its dual, $d_L(C^\perp),$ is at least three.

\section{Some useful identities} \label{identities}

In this section, we derive several identities related to linear codes over $\ZZ_4.$  We will use these identities to obtain parameters and weight distributions of Plotkin-optimal two-weight projective codes over $\ZZ_4$.

Let $G \in \ZZ_4^{(k_1+k_2)\times n}$ be a generator matrix of a linear code $C$ over $\ZZ_4.$ Following Kl{\o}ve \cite{Klove}, for any $\cv\in\ZZ_4^{k_1+k_2},$ we define the \emph{multiplicity} of $\cv,$ denoted by $\mu(\cv),$ as the number of occurrences of $\cv$ as a column vector in $G.$ Observe that for any $\xv\in\ZZ_4^k$ with $k:=k_1+k_2$, we have
\[
w_L\left(\xv G\right)=\sum_{\cv\in\ZZ_4^k} \mu(\cv) w_L(\xv\cdot\cv).
\]
It is also clear that
\[
\sum_{\cv\in\ZZ_4^k} \mu(\cv)=n.
\]
It is known that for any nonzero codeword $\cv\in\ZZ_4^k$, the sum $\sum_{\xv\in\ZZ_4^k}w_L(\xv\cdot \cv)$ is constant, and depends only on $k.$

\begin{lem}[\text{\cite[Lemma 3.1]{Hopein1}}]\label{lemma1}
	Let $k$ be a positive integer. If $\cv=(c_1,c_2,\dots,c_k)\in\ZZ_4^k$ is a nonzero vector, then
	\[
	\sum_{\xv\in\ZZ_4^k}w_L(\xv\cdot \cv)=4^k.
	\]
\end{lem}

Moreover, we have a constant sum of Lee weights of codewords in the linear code as follows.

\begin{lem}[\text{\cite[Lemma 3.2]{Hopein1},\cite[Lemma 3.3]{Hopein1}}]\label{Constant summation}
	Let $C$ be a linear code over $\ZZ_4$ with parameters $[n,4^{k_1}2^{k_2},d_L]$. If $G$ is a generator matrix of $C$ and $\mu(\mathbf{0})=0$, then
	\[
	 \sum_{\xv\in \ZZ_4^{k_1+k_2}} w_L(\xv G)=4^{k_1+k_2} n \quad\text{ and }\quad\sum_{\cv\in C} w_L(\cv)=|C|n.
	\]
\end{lem}

Now, consider the quotient group $\ZZ_4^k/(2\ZZ_4)^k$ for a positive integer $k.$ For any $\mathbf{S}\in\ZZ_4^k/(2\ZZ_4)^k,$ it is clear that $|\mathbf{S}|=|(2\ZZ_4)^k|=2^k$.  Moreover, for any $\xv,\yv\in\mathbf{S}$, we have $2\xv=2\yv.$  We can prove the identity stronger than the one in Lemma \ref{lemma1} and Lemma \ref{Constant summation} by proving that $\sum_{\xv\in \mathbf{S}} w_L(\xv\cdot\cv)$ and $\sum_{\xv\in \mathbf{S}} w_L(\xv G)$ are constant for any $\mathbf{S}\in \ZZ_4^k/(2\ZZ_4)^k$ in certain conditions.  We first consider the case when $\mathbf{S}=(2\ZZ_4)^k.$  For this purpose, we recall the property of one-weight codes below.

\begin{theorem}[\text{\cite[Theorem 4.5]{Shi2014},\cite[Theorem 4.2]{Hopein1}}]\label{opt}
    For any nonnegative integer $k_1,k_2$ with $2k_1+k_2\geq 1,$ let $G^{(k_1,k_2)}$ denote the generator matrix of a linear code $C^{(k_1,k_2)}$ whose columns consist of all possible nonzero vectors in $\ZZ_4^{k_1}\times (2\ZZ_4)^{k_2}.$ Then $C^{(k_1,k_2)}$ is a one-weight linear code over $\ZZ_4$ with parameters $[n,4^{k_1}2^{k_2},d_L],$ where $n=4^{k_1}2^{k_2}-1,$ and $d_L=4^{k_1}2^{k_2}.$
\end{theorem}

By using the above theorem, we obtain the following lemma.

\begin{lem}\label{id1}
	Let $k$ be a positive integer. If $\cv=(c_1,c_2,\dots,c_k)\in\ZZ_4^k$, then
	\[
	\sum_{\xv\in(2\ZZ_4)^k}w_L(\xv\cdot \cv)=
	\begin{cases}
	0 &\text{if } \cv\in(2\ZZ_4)^k,\\
	2^k &\text{if } \cv\in\ZZ_4^k-(2\ZZ_4)^k.\\
	\end{cases}
	\]
\end{lem}

\begin{proof}
Since it is clear for $\cv\in(2\ZZ_4)^k,$ we prove the case $\cv\in\ZZ_4^k-(2\ZZ_4)^k.$  By Theorem \ref{opt}, we know that the code $C^{(0,k)}$ has parameters  $[2^k-1,4^0 2^k, 2^k]$ and also a one-weight code having weight $2^k.$ Because $\cv G^{(0,k)}$ is a nonzero codeword in $C^{(0,k)},$ we have
\[
2^k=w_L(\cv G^{(0,k)})=\sum_{\xv\in\ZZ_4^k} \mu(\xv) w_L(\xv\cdot\cv)=\sum_{\xv\in(2\ZZ_4)^k}w_L(\xv\cdot \cv).
\qedhere\]
\end{proof}

A constant summation of the Lee weight along any nonzero element in $\ZZ_4^k/(2\ZZ_4)^k$ is given in the following. Throughout the rest of this paper, the notation $[a,b]_\ZZ$ denotes the set $\{x\in\ZZ\;\colon\;x\geq a\text{ and }x\leq b\}$, where $a,b$ are integers.

\begin{lem}\label{id2}
    Let $k$ be a positive integer and let $\mathbf{S}$ be a nonzero element of $\ZZ_4^k/(2\ZZ_4)^k.$ If $\cv=(c_1,c_2,\dots,c_k)\in\ZZ_4^k-(2\ZZ_4)^k$, then
	\[
	\sum_{\xv\in \mathbf{S}}w_L(\xv\cdot \cv)=2^k.
	\]
\end{lem}

\begin{proof}
Let $\xv\in \mathbf{S}.$ We consider two cases.

\underline{Case 1}: If $w_L(\xv\cdot\cv)=1,$ then for any $\yv\in\mathbf{S}$ we have $2\yv\cdot\cv=2\xv\cdot\cv=2$, which implies $w_L(\yv\cdot\cv)=1.$ Hence, $\sum_{\xv\in \mathbf{S}}w_L(\xv\cdot \cv)=|\mathbf{S}|=2^k.$

\underline{Case 2}: If $w_L(\xv\cdot\cv)\in\{0,2\},$ take $j\in[1,k]_\ZZ$ such that  $c_j\in\{1,3\}.$  Without loss of generality, assume $w_L(\xv\cdot\cv)=0.$ For $i\in[1, k]_\ZZ$, define a vector $\ev_i=(e_{i1},e_{i2},\dots e_{ik})\in\ZZ_4^k$ as
\[
e_{ij}:=
\begin{cases}
1, & \text{if }j=i,\\
0, & \text{if }j\neq i.
\end{cases}
\]
Now, define a map $f : \ZZ_4^k \longrightarrow \ZZ_4^k$ by $f(\xv)=\xv+2\ev_i.$ It is easy to check that $f$ is a bijection from $A:=\{\xv\in\mathbf{S}\,\colon\,w_L(\xv\cdot\cv)=0\}$ to $B:=\{\xv\in\mathbf{S}\,\colon\,w_L(\xv\cdot\cv)=2\}.$  For any  $\yv\in\mathbf{S}$, we have $2\yv\cdot\cv=2\xv\cdot\cv=0,$ so $w_L(\yv\cdot\cv)\in\{0,2\}$. Hence, $A\cup B=\textbf{S}$ and
$|A|=|B|=\frac{\mathbf{|S|}}{2}=2^{k-1}.$
We conclude that
\[
\sum_{\xv\in \mathbf{S}}w_L(\xv\cdot \cv)=\sum_{\xv\in A}w_L(\xv\cdot \cv)+\sum_{\xv\in B}w_L(\xv\cdot \cv)=2^k.
\qedhere\]
\end{proof}

From the constant summation properties as given in Lemma \ref{id1} and Lemma \ref{id2}, we obtain two new identities similar to Lemma \ref{Constant summation}.

\begin{cor}\label{Constant summation 2}
    Let $C$ be a linear code over $\ZZ_4$ with parameters $[n,4^{k_1}2^{k_2},d_L]$ and a generator matrix $G.$ If $\mu(\cv)=0,$ for all $\cv\in(2\ZZ_4)^{k_1+k_2},$  then
	\begin{equation}\label{cs1}
	    \sum_{\xv\in (2\ZZ_4)^{k_1+k_2}-\{\mathbf{0}\}} w_L(\xv G)=2^{k_1+k_2} n.
	\end{equation}
\end{cor}

\begin{proof}
Let $k:=k_1+k_2$. By Lemma \ref{id1}, we have
\[
\begin{aligned}
\sum_{\xv\in (2\ZZ_4)^{k}} w_L(\xv G) &=	\sum_{\xv\in (2\ZZ_4)^{k}}\sum_{\cv\in \ZZ_4^k}\mu(\cv)w_L(\xv\cdot \cv)=\sum_{\cv\in \ZZ_4^k}\mu(\cv)\left(\sum_{\xv\in (2\ZZ_4)^{k}}w_L(\xv\cdot \cv)\right)=2^{k} n.
\end{aligned}
\]
The result follows as $w_L(\mathbf{0}G)=0.$
\end{proof}

Similarly, we can prove the following using Lemma \ref{id2}.

\begin{cor}\label{Constant summation 3}
    Let $C$ be a linear code over $\ZZ_4$ with parameters $[n,4^{k_1}2^{k_2},d_L]$ and a generator matrix $G$. If  $\mu(\cv)=0,$  for all $\cv\in(2\ZZ_4)^{k_1+k_2},$  and $\mathbf{S}$ is a nonzero element in $\ZZ_4^{k_1+k_2}/(2\ZZ_4)^{k_1+k_2}$, then
	\begin{equation}\label{cs2}
	    \sum_{\xv\in \mathbf{S}} w_L(\xv G)=2^{k_1+k_2} n.
	\end{equation}
\end{cor}
Equations (\ref{cs1}) and (\ref{cs2}) will be very useful to prove several results in next section.

\section{Two-weight projective codes over $\ZZ_4$}
In \cite{Shi2018, Shi2014, Shi2017, Shi2020}, Shi and his coauthors considered two-weight projective codes.  In this section, we provide all possible parameters and weight distributions of Plotkin-optimal two-weight projective codes over $\ZZ_4.$  As consequences, we show that the results obtained before in \cite{Shi2018, Shi2014, Shi2017, Shi2020} are just special cases of our results.

First, let us recall the necessary and sufficient existence conditions for projective codes.

\begin{lem}[\text{\cite[Lemma 4.8]{Shi2020}}]\label{projective}
    Let $C$ be a linear code over $\ZZ_4$ with type $4^{k_1}2^{k_2}$ of length $n$ and a generator matrix $G$. Then $C$ is projective if and only if the following two conditions hold:
    \begin{enumerate}
        \item[(1)] every column of $G$ contains 1 or 3 (i.e. $\mu(\cv)=0,$ for all $\cv\in (2\ZZ_4)^{k_1+k_2}$);

        \item[(2)] any two columns of G are not multiple of each other by $\pm 1$ (i.e. $\mu(\cv)+\mu(3\cv)\leq 1,$ for all $\cv\in\ZZ_4^{k_1}\times(2\ZZ_4)^{k_2},$ given $G$ in a standard form).
    \end{enumerate}
\end{lem}
By using Lemma \ref{projective}, we can easily obtain an upper bound for the length of projective codes.

\begin{prop}\label{length}
If $C$ is a projective code of length $n$ over $\ZZ_4$ with type $4^{k_1}2^{k_2}$, then $n\leq 2^{2k_1+k_2-1}-2^{k_1+k_2-1}$.
\end{prop}

\begin{proof}
Let $k:=k_1+k_2$. From Lemma \ref{projective}, it is clear that
\[
\begin{aligned}
2n=2\sum_{\cv\in\ZZ_4^{k}}\mu(\cv)&=\sum_{\cv\in\ZZ_4^{k}}\mu(\cv)+\mu(3\cv)\leq \left|\left(\ZZ_4^{k_1}\times(2\ZZ_4)^{k_2}\right)-(2\ZZ_4)^k\right|=2^{2k_1+k_2}-2^{k_1+k_2}.
\end{aligned}
\]
Here, we consider the multiplicity using the generator matrix in a standard form.
\end{proof}

We knew that Plotkin-optimal projective codes always satisfy $d_L=n.$
From the above proposition, it is clear that $n<|C|-1$. Thus, a projective code $C$ of length $n$ is Plotkin-optimal if and only if
\[
d_L(C)=\left \lfloor \frac{|C|}{|C|-1}n \right \rfloor=\left \lfloor n+\frac{n}{|C|-1} \right \rfloor=n.
\]
As a consequence, Plotkin-optimal projective codes always contain a codeword of Lee weight $n.$  Hence, for two-weight codes, we can obtain the other weight by using the result in \cite{Shi2014} below.

\begin{prop}[\text{\cite[Proposition 5.1 (12)]{Shi2014}}]\label{weight-PC}
Let $C$ be any linear code of length $n$ over $\ZZ_4$. If $C$ is a two-Lee weight projective code with nonzero weights $w_1$ and $w_2$, then
\begin{equation}\label{proj-id}
    n(2n+1)-(w_1+w_2)2n+2w_1w_2\left(1-\frac{1}{|C|}\right)=0.
\end{equation}
\end{prop}

From Proposition \ref{weight-PC}, we conclude that the other weight of Plotkin-optimal projective codes must be $|C|/2.$  The relation between two-weight projective codes and Plotkin-optimal codes is summarized in the following corollary.
\begin{cor}\label{summary}
    Let $C$ be any linear code of length $n$ over $\ZZ_4$. If $C$ is a two-Lee weight projective code with nonzero weights $w_1<w_2$, then the following statements are equivalent:
    \begin{itemize}
        \item[(1)] $C$ is Plotkin-optimal;
        \item[(2)] $w_1=n$;
        \item[(3)] $\displaystyle w_2=\frac{|C|}{2}$.
    \end{itemize}
\end{cor}

Notice that $w_1=n$ and $w_2=\frac{|C|}{2}$ satisfy Equation (\ref{proj-id}) for all positive integers $n<\frac{|C|}{2}.$ In this case, we could not determine the nonexistence of Plotkin-optimal two-weight projective codes of certain length from that equation.  Fortunately, by using additional equations, namely Equations (\ref{cs1}) and (\ref{cs2}), we obtain the nonexistence results for many values of $n$ as follows (note that Corollary \ref{Constant summation 2} and \ref{Constant summation 3} hold for any projective codes by Lemma \ref{projective}).

\begin{cor}\label{constant weight}
     Let $C$ be a two-Lee weight projective code of length $n$ over $\ZZ_4$ with a generator matrix $G.$ If $C$ is Plotkin-optimal, then $w_L(\xv G)=n,$ for all $\xv\in\ZZ_4^{k_1+k_2}-(2\ZZ_4)^{k_1+k_2}.$
\end{cor}

\begin{proof}
Let $\mathbf{S}$ be any nonzero element of $\ZZ_4^{k_1+k_2}/(2\ZZ_4)^{k_1+k_2}$. From Lemma \ref{Constant summation 3}, we have
\[
2^{k_1+k_2} n=\sum_{\cv\in \mathbf{S}} w_L(\cv G)\geq |\mathbf{S}|d_L=2^{k_1+k_2}n.
\]
The equality occurs when $w_L(\cv G)=n,$ for all $\cv\in \mathbf{S}.$  Since $\mathbf{S}$ is arbitrary, $w_L(\xv G)=n,$ for all $\xv\in\ZZ_4^{k_1+k_2}-(2\ZZ_4)^{k_1+k_2}.$
\end{proof}

Let $A_w$ denote the number of codewords in $C$ having the Lee weight equal to $w.$  The theorem below shows possible parameters of Plotkin-optimal two-Lee weight projective codes over $\ZZ_{4}.$

\begin{theorem}\label{main-theo}
    Let $C$ be a Plotkin-optimal two-weight projective code with type $4^{k_1}2^{k_2}$ of length $n$ over $\ZZ_4$. Then $n=d_L=2^{2k_1+k_2-1}-2^{2k_1+k_2-t-1},$ for some integer $1\leq t\leq k_1$. Moreover, the nonzero weights of $C$ are $w_1=2^{2k_1+k_2-1}-2^{2k_1+k_2-t-1}$ and $w_2=2^{2k_1+k_2-1}$, with $A_{w_1}=2^{2k_1+k_2}-2^t$ and $A_{w_2}=2^t-1$.
\end{theorem}

\begin{proof}
	Let $G$ be a generator matrix of $C.$ Let $w_1<w_2$ be the two nonzero weights of $C.$  By Corollary \ref{summary}, we have $w_1=n$ and $w_2=\frac{|C|}{2}.$
	
	Let $s$ be the number of vectors $\xv\in(2\ZZ_4)^k$ such that $w_L(\xv G)=\frac{|C|}{2}.$  Since $C$ is projective, $C$ satisfies Corollary \ref{Constant summation 2}:
	\[
	\begin{aligned}
	2^{k_1+k_2}n&=\sum_{\cv\in (2\ZZ_4)^{k_1+k_2}-\{\mathbf{0}\}} w_L(\cv G)=s\frac{|C|}{2}+(2^{k_1+k_2}-1-s)n.
	\end{aligned}
	\]
    From the above equation, we have $n=\frac{s}{2(s+1)}|C|.$ It implies that $s|C|$ is divisible by $s+1,$ and hence  $s+1$ divides $|C|=2^{2k_1+k_2}.$ It forces $s+1=2^t$ for some $t\in\ZZ^+.$ By Proposition \ref{length}, we obtain
	\[
	\begin{aligned}
	2^{2k_1+k_2-1}-2^{k_1+k_2-1}\geq n &=\frac{2^t-1}{2^{t+1}}2^{2k_1+k_2}=2^{2k_1+k_2-1}-2^{2k_1+k_2-t-1},
	\end{aligned}
	\]
    which implies $1\leq t\leq k_1.$  We conclude that \[n=2^{2k_1+k_2-1}-2^{2k_1+k_2-t-1}=w_1
    \text{ and } w_2=\frac{|C|}{2}=2^{2k_1+k_2-1}
    \]
    for some integer $t$ with $1\leq t\leq k_1.$ Now, to find the values of $A_{w_1}$ and $A_{w_2},$ we know from Corollary \ref{constant weight}, that $w_L(\xv G)=n,$ for all $\xv\in\ZZ_4^{k_1+k_2}-(2\ZZ_4)^{k_1+k_2}.$ Therefore, $A_{w_1}=|C|-1-s=2^{2k_1+k_2}-2^t$ and $A_{w_2}=s=2^t-1.$
\end{proof}

In the remark below, we compare our result above (Theorem \ref{main-theo}) with the previous results available in \cite{Shi2018, Shi2014, Shi2017, Shi2020}.
\begin{rem}
	\begin{enumerate}
		\item Any code in the family of two-weight projective codes in \cite[Theorem 5.6]{Shi2014} (same as \cite[Theorem 3.4]{Shi2017}) is Plotkin-optimal and also a special case of Theorem \ref{main-theo} with $t=1.$ Moreover, \cite[Theorem 4.2]{Shi2018} shows that any two-weight projective code with $d_L(C^\perp)\geq 4$ must be Plotkin-optimal and also a special case of Theorem \ref{main-theo} with $t=1.$
		
		\item Any code in two families of two-weight projective codes in \cite[Theorem 4.2]{Shi2020} and \cite[Corollary 4.5]{Shi2020} is Plotkin-optimal and also a special case of Theorem \ref{main-theo} with $t=k_1-1$ and $t=k_1.$
	\end{enumerate}
\end{rem}

Now, let us consider the Gray image of any Plotkin-optimal two-weight projective code over $\ZZ_4$. First, recall the parameters of SU1-type two-weight $[n,k]$ linear codes over $\FF_q$ (see Figure 1a and Figure 2a in \cite{Calderbank1986}):
\[
n=\frac{q^l-q^m}{q-1}, ~k=l,  ~w_1=q^{l-1}-q^{m-1}, ~w_2=q^{l-1}, ~A_{w_1}=q^l-q^{l-m}, ~A_{w_2}=q^{l-m}-1,
\]
for some integers $l> 1$ and $1\leq m\leq l-1.$ From Theorem \ref{main-theo}, the Gray image of any Plotkin-optimal two-weight projective code $C$ over $\ZZ_4$ is a binary code with parameters
\[n=2n(C)=2^{2k_1+k_2}-2^{2k_1+k_2-t}, k=2k_1+k_2, w_1=2^{2k_1+k_2-1}-2^{2k_1+k_2-t-1}, w_2=2^{2k_1+k_2-1}\]
and weight distribution $A_{w_1}=2^{2k_1+k_2}-2^t$, $A_{w_2}=2^t-1$. Notice that $\Phi(C)$ is not necessarily linear, but it always has the same parameters and weight distribution as some binary two-weight linear codes of type SU1. In this case, $q=2, l=2k_1+k_2$, and $m=2k_1+k_2-t$.

Moreover, we can also conclude that the syndrome graph (see \cite{Shi2018} for definition) of the dual of any Plotkin-optimal two-weight projective code over $\ZZ_4$ is a strongly regular graph on $2^{2k_1+k_2}$ vertices, degree $2^{2k_1+k_2}-2^{2k_1+k_2-t}$, and restricted eigenvalues $0$ and $-2^{2k_1+k_2-t}$. This graph is also a complete multipartite graph of type $K_{a\times m}$, where $a=2^t$ and $m=2^{2k_1+k_2-t}$. The details are very similar to \cite[Section 6]{Shi2020}.

\section{A general family}
In this section, we first recall a family of free Plotkin-optimal two-weight projective codes over $\ZZ_4$ constructed in \cite{Shi2020}. Using this family of codes, we prove the existence of Plotkin-optimal two-weight projective codes over $\ZZ_4$ with the same parameters and weight distribution as in Theorem \ref{main-theo}, for any $1\leq t\leq k_1$.

\subsection{A family of two weight codes with $k_2=0$}
Let $k$ be a positive integer. Following Shi et al. \cite{Shi2020}, we define a matrix $Y_k$ as follows: for $k=1$, we define $Y_k:=
\begin{bmatrix}
	1
\end{bmatrix};$ and for $k\geq 2,$ we define
\[Y_{k}:=
\arraycolsep=3.0pt
\left[\begin{array}{@{}c|c|c|c| c@{}}
	Y_{k-1} & Y_{k-1} & Y_{k-1} & Y_{k-1} & B_{k-1} \\\hline
	0\;0\dots 0 &1\;1\dots 1 &2\;2\dots 2 &3\;3\dots 3 &1\;1\dots 1
\end{array}\right],
\]
where $B_{k-1}$ is a $(k-1)\times 2^{k-1}$ matrix whose columns are all possible vectors in $(2\ZZ_4)^{k-1}$. From the definition, it is easy to check that the linear code generated by $Y_k$ is projective for any positive integer $k$ (Lemma \ref{projective}). In other words, $\mu(\cv)=0,$ for all $\cv\in (2\ZZ_4)^{k}$ and $\mu(\cv)+\mu(3\cv)\leq 1,$ for all $\cv\in\ZZ_4^{k}.$ Now, we reprove \cite[Proposition 3.1]{Shi2020} and \cite[Corollary 3.2]{Shi2020} with a slightly different (and shorter) proof.
\begin{prop}[\text{\cite[Proposition 3.1]{Shi2020}, \cite[Corollary 3.2]{Shi2020}}]\label{prop-Yk}
If $C$ is the code generated by $Y_k$, then $C$ is a free two-weight code of length
$n = 2^{k-1}(2^k-1)$, with nonzero Lee weights $w_1 = 2^{k-1}(2^k-1)$ and $w_2 = 2^{2k-1},$ also $A_{w_1}= 4^k - 2^k$ and $A_{w_2}= 2^k - 1.$ Moreover,
\[
w_L(\xv Y_k)=
\begin{cases}
    2^{2k-1}, &\text{if } \xv\in(2\ZZ_4)^k-\{\mathbf{0}\},\\
	2^{k-1}(2^k-1), &\text{if } \xv\in\ZZ_4^k-(2\ZZ_4)^k.\\
\end{cases}
\]
\end{prop}
\begin{proof}
It is easy to check that $n = 2^{k-1}(2^k-1), k_1=k$, and $k_2=0$. From the proof of Proposition \ref{length}, we conclude that $\mu(\cv)+\mu(3\cv)=1$ for all $\cv\in\ZZ_4^{k}$. Now, for any $\xv\in \ZZ_4^k-\{\mathbf{0}\}$ we have
\begin{align*}
    w_L(\xv Y_k)&=\frac{1}{2}\left(2\sum_{\cv\in\ZZ_4^k} \mu(\cv) w_L(\xv\cdot\cv)\right)\\
    &=\frac{1}{2}\left(\sum_{\cv\in\ZZ_4^k} \mu(\cv) w_L(\xv\cdot\cv)+\sum_{\cv\in\ZZ_4^k} \mu(3\cv) w_L(\xv\cdot3\cv)\right)\\
    &=\frac{1}{2}\left(\sum_{\cv\in\ZZ_4^k}(\mu(\cv)+\mu(3\cv))w_L(\xv\cdot\cv)\right)\\
    &=\frac{1}{2}\left(\sum_{\cv\in\ZZ_4^k}w_L(\xv\cdot\cv)-\sum_{\cv\in(2\ZZ_4)^k}w_L(\xv\cdot\cv)\right)\\
    &=\begin{cases}
    2^{2k-1}, &\text{if } \xv\in(2\ZZ_4)^k-\{\mathbf{0}\},\\
	2^{k-1}(2^k-1), &\text{if } \xv\in\ZZ_4^k-(2\ZZ_4)^k.\\
    \end{cases}
\end{align*}
The last equality is obtained from Lemma \ref{lemma1} and Lemma \ref{id1}. It is clear that $A_{w_1}=|\ZZ_4^k-(2\ZZ_4)^k|=4^k-2^k$ and $A_{w_2}=|(2\ZZ_4)^k-\{\mathbf{0}\}|=2^k-1$.
\end{proof}

\begin{rem}
Note that $"\xv \in(2\ZZ_4)^k-\{\mathbf{0}\}"$ and $"\xv\in\ZZ_4^k-(2\ZZ_4)^k"$ are equivalent to "the order of $\cv=\xv Y_k$ is 2" and "the order of $\cv=\xv Y_k$ is 4" in \cite[Corollary 3.2]{Shi2020}, respectively.  Here, the order of a nonzero codeword $\cv$ is the smallest positive integer $k$ such that $k \cv=\mathbf{0}.$
\end{rem}

Notice that the code generated by $Y_k$ is Plotkin-optimal, because
\[
\left\lfloor\frac{|C|}{|C|-1}n \right\rfloor=\left\lfloor\frac{4^k}{4^k-1}\left(2^{k-1}(2^k-1)\right) \right\rfloor=2^{k-1}(2^k-1)=d_L(C).
\]
Moreover, we can prove the uniqueness (up to equivalence) of free Plotkin-optimal two-weight projective codes over $\ZZ_4$ of length $n = 2^{k-1}(2^k-1)$. From the proof of Proposition \ref{length}, notice that $n$ attains the upper bound only if $\mu(\cv)+\mu(3\cv)=1,$ for all $\cv\in\ZZ_4^{k}.$ It is easy to see that any two codes with $\mu(\cv)+\mu(3\cv)=1,$ for all $\cv\in\ZZ_4^{k},$ are equivalent. Note that we can replace $C$ in Proposition $\ref{prop-Yk}$ with other equivalent codes.

\subsection{An infinite family of two-weight codes}
We begin with the following construction.
\begin{theorem}\label{konstruk}
    Let $C$ be a two-weight linear code of parameters $[n,4^{k_1}2^{k_2},d_L]$ with a generator matrix $G$ and nonzero weights $w_1<w_2$. Let $C_1$ and $C_2$ be two linear codes generated by the matrices
    \begin{equation}\label{matrix1}
	G_1:=\arraycolsep2.5 pt \left[\begin{array}{@{}c|c|c|c@{}}
		G & G & G & G \\
		\hline
		0\;0\dots 0 &1\;1\dots 1 &2\;2\dots 2 &3\;3\dots 3
	\end{array}\;\right]
	\end{equation}
	and
	\begin{equation}\label{matrix2}
	    G_2:=\arraycolsep2.5 pt\left[\begin{array}{@{}c|c@{}}
		G & G \\
		\hline
		0\;0\dots 0  &2\;2\dots 2
	    \end{array}\;\right],
	\end{equation}
    respectively. Then the following five statements hold:
    \begin{enumerate}
        \item[(1)] $C$ is projective if and only if $C_1$ and $C_2$ are projective;

        \item[(2)] $A_{4n}(C_1)=3|C|+\delta_{n,w_1}A_n(C);\text{ and }A_{4w}(C_1)=A_w(C),$ for any integer $w \in \{w_1,w_2\}-\{n\};$

        \item[(3)] $A_{2n}(C_2)=|C|+\delta_{n,w_1}A_n(C); \text{ and }A_{2w}(C_2)=A_w(C),$ for any integer $w \in \{w_1,w_2\}-\{n\};$

        \item[(4)] $C_1$ and $C_2$ are two-weight projective codes if and only if $C$ is Plotkin-optimal and projective;

        \item[(5)] If $C$ is Plotkin-optimal, then $C_1$ and $C_2$ is also Plotkin-optimal.
    \end{enumerate}
\end{theorem}
\begin{proof}
\begin{enumerate}
    \item[(1)] Straightforward from Lemma \ref{projective}.

    \item[(2)] For any $\cv^\ast\in C_1$, there exists $\xv=(x_1,x_2,\dots, x_{k+1})\in \ZZ_4^{k+1} (k:=k_1+k_2)$ such that $\cv^\ast=\xv G_1$. For any $\yv=(y_1,y_2,\dots, y_{k+1})\in \ZZ_4^{k+1}$, let $\yv'$ denote the vector $(y_1,y_2,\dots,y_k)\in\ZZ_4^{k}$. Notice that
    \[
    w_L(\cv^\ast)=w_L\left(\xv G_1\right)=\sum_{\cv\in\ZZ_4^{k+1}} \mu(\cv) w_L(\xv\cdot\cv)=\sum_{\cv'\in\ZZ_4^{k}} \mu(\cv') \sum_{c_{k+1}=0}^3 w_L(\xv'\cdot\cv'+x_{k+1}c_{k+1}).
\]
Note that $\mu(\cv)$ denotes the multiplicity of $\cv$ in $G_1$ for $\cv\in\ZZ_4^{k+1}$ and $\mu(\cv')$ denotes the multiplicity of $\cv'$ in $G$ for $\cv'\in\ZZ_4^{k}$. We obtain the following:
\begin{itemize}
    \item If $x_{k+1}=0$, then $w_L(\cv^\ast)=4w_L(\xv'G)$ where  $\xv'G\in C$.
    \item If $x_{k+1}\in\{1, 2, 3\}$, then
    $w_L(\cv^\ast)=4n$.
\end{itemize}
In this case, $A_{4n}(C_1)=3|C|+\delta_{n,w_1}A_n(C)\text{ and }A_{4w}(C_1)=A_w(C)$ for any integer $w \in \{w_1,w_2\}-\{n\}.$

\item[(3)] Similar to (2).

\item[(4)] If $C_1$ and $C_2$ are two-weight projective codes, then $C$ is projective from (1). Suppose that $C$ is not Plotkin-optimal. Then $d_L(C)=w_1<n$. If $w_2\leq n$, then from Lemma \ref{Constant summation} we have
\[|C|n=\sum_{\cv\in C} w_L(\cv)<(|C|-1)n,\]
which is impossible. Thus, $w_2>n$. From (2) and (3), we have $A_{4w}>0$ and $A_{2w}>0$ for $w\in\{w_1,n,w_2\}$. Hence, $C_1$ and $C_2$ are not two-weight codes, a contradiction. Therefore, $C$ must be Plotkin-optimal.

If $C$ is Plotkin-optimal and projective, then from (1), $C_1$ and $C_2$ are projective. Moreover, by Corollary \ref{summary}, we have $w_1=n$ and $w_2=\frac{|C|}{2}$. From (2) and (3), we conclude that the nonzero weights of $C_1$ must be $4n$, $2|C|$ and the nonzero weights of $C_2$ must be $2n$, $|C|$. Therefore, $C_1$ and $C_2$ are two-weight codes.

\item[(5)] Straightforward from (2), (3), (4), and Corollary \ref{summary}. \qedhere
\end{enumerate}
\end{proof}

We can apply the construction method as mentioned in Theorem \ref{konstruk} finitely many times to obtain families of two-weight codes.

\begin{theorem}\label{proj-konstruk}
    Let $C$ be a Plotkin-optimal two-weight projective code over $\ZZ_4$ with parameters $[n, 4^{k_1}2^{k_2}, n]$ and nonzero Lee weights $m,n.$ Let $a,b$ be nonnegative integers, not both are zero. Then there exists a Plotkin-optimal two-weight projective code $C'$ over $\ZZ_4$ with parameters $[4^a 2^b n,$$4^{k_1+a}2^{k_2+b}, 4^a 2^b n].$ Moreover, $w_1=4^a 2^b m$, $w_2= 4^a 2^b n$ and $A_{w_1}(C')=A_m(C)$, $ A_{w_2}(C')=4^{k_1+a}2^{k_2+b}-1-A_m(C).$
\end{theorem}

\begin{proof}
Let $G$ be the generator matrix of $C$ (in a standard form). Consider a linear code $C'$ over $\ZZ_4$ with a generator matrix $G'$ whose columns consist of all distinct nonzero vectors
\[
(c_1,\dots,c_{k_1+k_2},c_{k_1+k_2+1},\dots,c_{k_1+k_2+a}, c_{k_1+k_2+a+1},\dots,c_{k_1+k_2+a+b})^T,
\]
where $(c_1,\dots,c_{k_1+k_2})^T$ is a column of $G$, $c_i\in \ZZ_4,$ for $i\in[k_1+k_2+1,k_1+k_2+a]_\ZZ$, and $c_i\in 2\ZZ_4,$ for $i\in[k_1+k_2+a+1,k_1+k_2+a+b]_\ZZ.$

The result follows by applying Theorem \ref{konstruk} $a+b$ times.
\end{proof}

As an illustration of Theorem \ref{proj-konstruk}, let us consider the example below.

\begin{ex}\label{C-proj-konstruk}
	Following the notation in Theorem \ref{proj-konstruk}, let $a=1$ and $b=2.$ If $C$ is a Plotkin-optimal two-weight projective code with parameters $[n, 4^{k_1}2^{k_2}, n]$ generated by a matrix $G,$ then by Theorem \ref{konstruk}(4), the linear code $C_1$ with generator matrix
	\[
	G_1=\arraycolsep2.5 pt
	\left[
	\begin{array}{@{}c|c|c|c@{}}
		G & G & G & G \\
		\hline
		0\;0\dots 0 &1\;1\dots 1 &2\;2\dots 2 &3\;3\dots 3\\
	\end{array}\;
	\right]
	\]
	is a Plotkin-optimal two-weight projective code with parameters $[4n, 4^{k_1+1}2^{k_2}, 4n]$ and nonzero Lee weights $4m,4n$. Moreover, by
	Theorem \ref{konstruk}(2), $A_{4m}(C_1)=A_m(C)\text{ and } A_{4n}(C_1)=3|C|+\delta_{n,n}A_n(C)=3|C|+A_n(C)=4|C|-1-A_m(C)=4^{k_1+1}2^{k_2}-1-A_{m}(C)$.
	
	Next, notice that the linear code $C'$ over $\ZZ_4$ with a generator matrix $G'$ whose columns consist of all distinct nonzero vectors
\[
(c_1,\dots,c_{k_1+k_2},c_{k_1+k_2+1}, c_{k_1+k_2+2},c_{k_1+k_2+3})^T,
\]
where $(c_1,\dots,c_{k_1+k_2})^T$ is a column of $G$, $c_{k_1+k_2+1}\in \ZZ_4,$ and $c_{k_1+k_2+2},c_{k_1+k_2+3}\in 2\ZZ_4$ is permutation-equivalent to the linear code $C_2$ with generator matrix
	\begin{align*}
	    G_2&=
	\arraycolsep2.5 pt
		\left[
		\begin{array}{@{}c|c|c|c|c|c|c|c|@{}}
			G & G & G & G & G & G & G & G\\
			\hline
			0\;0\dots 0 &1\;1\dots 1 &2\;2\dots 2 &3\;3\dots 3 & 0\;0\dots 0 &1\;1\dots 1 &2\;2\dots 2 &3\;3\dots 3\\
			0\;0\dots 0 & 0\;0\dots 0 & 0\;0\dots 0 & 0\;0\dots 0 &
			2\;2\dots 2 & 2\;2\dots 2 & 2\;2\dots 2 & 2\;2\dots 2\\
			0\;0\dots 0 & 0\;0\dots 0 & 0\;0\dots 0 & 0\;0\dots 0 &
			0\;0\dots 0 & 0\;0\dots 0 & 0\;0\dots 0 & 0\;0\dots 0\\
		\end{array}\; \right. \\&\quad\quad
		\left.
		\arraycolsep2.5 pt
		\begin{array}{@{}c|c|c|c|c|c|c|c@{}}	
			G & G & G & G & G & G & G & G\\
			\hline
			0\;0\dots 0 &1\;1\dots 1 &2\;2\dots 2 &3\;3\dots 3 & 0\;0\dots 0 &1\;1\dots 1 &2\;2\dots 2 &3\;3\dots 3\\
			0\;0\dots 0 & 0\;0\dots 0 & 0\;0\dots 0 & 0\;0\dots 0 &
			2\;2\dots 2 & 2\;2\dots 2 & 2\;2\dots 2 & 2\;2\dots 2\\
			2\;2\dots 2 & 2\;2\dots 2 & 2\;2\dots 2 & 2\;2\dots 2 &
			2\;2\dots 2 & 2\;2\dots 2 & 2\;2\dots 2 & 2\;2\dots 2\\
		\end{array}\;
		\right]
		\\
		&=\arraycolsep2.5 pt
	\left[
	\begin{array}{@{}c|c|c|c@{}}
			G_1 & G_1 & G_1 & G_1\\
			\hline
			0\;0\dots 0 & 2\;2\dots 2 & 0\;0\dots 0 & 2\;2\dots 2\\
			0\;0\dots 0 & 0\;0\dots 0 & 2\;2\dots 2 & 2\;2\dots 2
 	\end{array}\;
	\right].
	\end{align*}
	From Theorem \ref{konstruk}(4), we conclude that $C'$ is a Plotkin-optimal two-weight projective code with parameters $[16n, 4^{k_1+1}2^{k_2+2}, 16n]$ and nonzero Lee weights $16m,16n$. Moreover, by applying
	Theorem \ref{konstruk}(3) twice, we have  $A_{16m}(C')=A_{16m}(C_2)=A_{4m}(C_1)=A_m(C)$ and $A_{16n}(C')=A_{16n}(C_2)=|C_2|-1-A_{16m}(C_2)=2^2|C_1|-1-A_{m}(C)=4^{k_1+1}2^{k_2+2}-1-A_{m}(C)$, as desired.
\end{ex}

\begin{rem}
The family of two-weight codes in \cite[Theorem 4.2]{Shi2020} is a special case of Theorem \ref{proj-konstruk} when $C$ is the linear code generated by $Y_{k_1},$ $a=0,$ and $b=k_2.$ The family of two-weight codes in \cite[Corollary 4.5]{Shi2020} is a special case of Theorem \ref{proj-konstruk} when $C$ is the linear code generated by $Y_{k_1-1},$ $a=1,$ and $b=k_2.$
\end{rem}

Finally, we provide a general family of two-weight codes, which includes all two-weight codes in two families constructed by Shi et al. \cite{Shi2020}.

\begin{theorem}\label{main-theo2}
    Let $k_1,k_2$ be nonnegative integers such that $2k_1+k_2\geq 2$. Then for any integer $t$ with $1\leq t\leq k_1,$ there exists a two-weight projective code with parameters $[n,4^{k_1}2^{k_2},d_L]$ where $n=d_L=2^{2k_1+k_2-1}-2^{2k_1+k_2-t-1}$. Moreover, the nonzero weights of the code are $w_1=2^{2k_1+k_2-1}-2^{2k_1+k_2-t-1}$ and $w_2=2^{2k_1+k_2-1}$, with $A_{w_1}=2^{2k_1+k_2}-2^t$ and $A_{w_2}=2^t-1$.
\end{theorem}

\begin{proof}
    By applying Theorem \ref{proj-konstruk} with $a=k_1-t$, $b=k_2$ on the linear code $C$ generated by $Y_t$ (the parameters of this code are given in Proposition \ref{prop-Yk}), we get the desired result.
\end{proof}
The Theorem \ref{main-theo2} above proves the existence of Plotkin-optimal two-weight projective codes over $\ZZ_4$ with the same parameters and weight distribution as the codes in Theorem \ref{main-theo} for any $1\leq t\leq k_1$.

\begin{rem}
If $C$ is a linear code over $\ZZ_4$ with a generator matrix $G$ whose columns consist of all distinct nonzero vectors\[
(c_1,\dots,c_{t},c_{t+1},\dots,c_{k_1}, c_{k_1+1},\dots,c_{k_1+k_2})^T,
\]
where $(c_1,\dots,c_{t})^T$ is a column of $Y_t$, $c_i\in \ZZ_4,$ for $i\in[t+1,k_1]_\ZZ$, and $c_i\in 2\ZZ_4,$ for $i\in[k_1+1,k_1+k_2]_\ZZ$, then we can completely determine the linearity of $\Phi(C)$. It is easy to check that $\Phi(C)$ is not linear for $1\leq t<k_1-1$ from \cite[Theorem 5]{Hammons1994}. The linearity of $\Phi(C)$ for $t=k_1-1$ and $t=k_1$ have been completely determined in \cite[Theorem 5.8]{Shi2020} and \cite[Theorem 5.6]{Shi2020}, respectively.   
\end{rem}

\section{Concluding remarks}
We constructed a general family of two-weight projective codes over $\ZZ_4.$  Our result generalized the previous results, mainly from \cite{Shi2020}. We also obtained all possible parameters and weight distributions of Plotkin-optimal two-Lee weight projective codes over $\ZZ_4$. However, the number of non-equivalent codes for some fixed parameters and weight distribution is still unknown, which is interesting to investigate further.

Although the two-weight projective codes obtained here are all Plotkin-optimal, it is not the case that all two-weight projective codes over $\ZZ_4$ are Plotkin-optimal. As an example, consider the linear code $C$ with generator matrix $G$ as follows:
\[
G=
\begin{bmatrix}
	1 &0 &0 &3 &3 &3 &2 &0 &0 &3 &1 &2 &3 &3\\
	0 &1 &0 &2 &1 &3 &0 &1 &3 &2 &0 &1 &1 &1\\
	0 &0 &1 &2 &2 &2 &3 &1 &1 &1 &1 &2 &1 &3
\end{bmatrix}.
\]
$C$ is a linear code with parameters $[14, 4^3 2^0, 12].$  The code $C$ is a two-weight code with nonzero weights $w_1=12,~w_2=16$ and weight distribution $A_{w_1}=28,~ A_{w_2}=35.$ It can be checked that $d_L(C^\perp)=3,$ and hence $C$ is a projective code. However, $d_L(C)=12 < \left \lfloor\frac{4^3}{4^3-1}14 \right \rfloor=14,$ and hence $C$ is not Plotkin-optimal. Moreover, the Gray image of this code has the same parameters and distribution as some two-weight binary linear codes of type SU2 (see Figure 1a and Figure 2a in \cite{Calderbank1986}). So far, we have found a family of non-Plotkin-optimal two-weight projective codes over $\ZZ_4$ of this type. It is interesting to find all possible parameters and weight distributions of two-weight linear projective codes over $\ZZ_4$ of this type and we are now working on that direction.

\section*{Acknowledgment}
The authors thank two anonymous referees for their meticulous reading of the manuscript.  Their suggestions have been very valuable in improving the presentation of this paper.

This research is supported by Institut Teknologi Bandung (ITB) and the Ministry of Education, Culture, Research, and Technology
(\emph{Kementerian Pendidikan, Kebudayaan, Riset, dan Teknologi (Kemdik\-budristek)}),
Republic of Indonesia.

\end{document}